\documentclass[letterpaper, 10 pt, conference]{ieeeconf}

\IEEEoverridecommandlockouts                              

\usepackage{nomencl}
\makeglossary
\usepackage{graphicx}
\usepackage{caption}
\usepackage{subcaption}
\usepackage{algorithm}
\usepackage{algorithmic}
\usepackage{amsmath}
\usepackage{graphicx}
\usepackage{enumerate}%
\usepackage{type1cm}
\usepackage{lettrine}
\usepackage{breqn}
\usepackage{graphicx}
\usepackage{amssymb}
\usepackage{framed}
\usepackage{tikz}
\usepackage{amsmath}
\usepackage{setspace}
 \usepackage{pgfplots}
\usepackage{caption}
\usepackage{subcaption}
\usepackage{mathtools}

\usepackage{epstopdf}

\usepackage{bm}

  \newcommand{\diag}{\mathop{\rm diag}}

\newtheorem{cond}{\textnormal{\textbf{Condition}}}

\newtheorem{problem}{\textnormal{\textbf{Problem}}}

\usepackage{calc}

\newcommand{\matr}[1]{\mathbf{#1}}

\newtheorem{theorem}{\textnormal{\textbf{Theorem}}}

\newtheorem{remark}{\textnormal{\textbf{Remark}}}
\newtheorem{proposition}{\textnormal{\textbf{Proposition}}}
\newcommand{\nunder}[2][5]{\mathrlap{\mkern\the\numexpr#1/2mu\relax\underlipe{\phantom{\mathrm{#2}\mkern-#1mu}}}#2}
\usepackage[font={footnotesize}]{caption}

\title{\LARGE \bf
Consensus-Based Torque Control of Deloaded Wind DFIGs for Distributed and Fair Dynamic Dispatching
}

\author{Stefanos~Baros$^{1}$ 
\thanks{$^{1}$Stefanos Baros is with the ECE department
of Carnegie Mellon University, Pittsburgh, PA, 15213 USA e-mail: {\tt sbaros@andrew.cmu.edu}}}

\begin{document}

\maketitle
\thispagestyle{empty}
\pagestyle{empty}

\begin{abstract}
In this paper we aim to address the problem of dynamically dispatching  a group of state-of-the-art deloaded wind generators (WGs) in a fair-sharing manner. We use the term dynamically since the WGs aim to dispatch themselves according to a varying committed WF power output. We first propose a leader-follower  protocol whose execution guarantees asymptotically, two control objectives. These are 1) reaching asymptotic consensus on the utilization level of all WGs and 2) the total power output of the WGs asymptotically converges to the reference value. Thereafter, we combine singular perturbation and Lyapunov theory to prove that, under certain conditions, the proposed protocol will asymptotically converge to its equilibrium. Finally, we derive a cooperative Control Lyapunov Function-based (CLF) controller for the rotor side converter (RSC) of each WG that realizes the protocol in practice. We demonstrate the effectiveness of our proposed protocol and the corresponding RSC controller design via simulations on the modified IEEE 24-bus RT system.

\end{abstract}


\IEEEpeerreviewmaketitle

\section{Introduction}
A recent study conducted by the  US Department Of Energy \cite{windvision} outlines the future of wind power in the US.  Specifically, it mentions that   10\% of the US electricity demand is expected to be produced by wind power by 2020, 20\% of the US electricity demand  by 2020 and  35\% of the US electricity demand  by 2050. In a similar status is Europe, where wind power integration is expected to increase significantly in the next years \cite{ackerman}. These studies evidence a recent world-wide tendency toward integrating a lot of wind power into power systems. On the other hand, integrating high levels of wind power into power systems raises an important challenge for those systems. That is, to maintain their stability, their reliability and their robustness \cite{ackerman}.
\par It is worthwhile realizing that in high-wind-integration settings, the WGs control will have a pronounced impact on the stability and performance of the power grids that accommodate the wind power.  For this reason, the ongoing regulations for the operation of WGs require the WGs to provide multiple advanced capabilities into the grid \cite{ackerman}. Between them,  frequency regulation, inertial response, power output smoothing, Low Voltage Ride-Through (LVRT) capability and voltage control \cite{ackerman}. Yet, by allowing communication between WGs,  capabilities that require coordination between WGs can be achieved.  Examples can be coordinated voltage control for regulation of the WF terminal voltage and coordinated power control with load-sharing between WGs. Such capabilities can be achieved efficiently with distributed control  methods that require  a limited number of dedicated communication links. 
\par In this paper, we focus on an important and  advanced capability that can be provided by a group of deloaded WGs. The capability for regulating their total power extracted from the wind such that it tracks a reference while they dynamically dispatch themselves in a fairly fashion. Deloaded WGs are characterized by their flexibility to increase their power output when they are commanded to. A WF that has its WGs operating in a deloaded regime aiming to provide the above capability has to address the following challenges. Firstly, to timely compute the reference power points of its WGs while taking into account  the local varying wind speed conditions. Secondly, to timely communicate these set-points to the WGs in order for the WF total power output to match the required reference. 

In the literature, most of the studied methods  relied  on centralized control schemes. Centralized schemes presume that information about the  wind speed, which is scattered throughout the WF, is obtained from the central controller. The central controller, having this information together with information regarding the total WF power output can then compute the power set-point for each WG individually. Then, it can dispatch the WGs accordingly. Finally, the power output of each WG can be regulated by its local controller such that the WG generates the reference power. An approach belonging to the above category can be found in \cite{constantpower}. Despite the fact that centralized-based approaches can address the problem discussed above they come with several drawbacks, rather critical to be neglected. Among others are, single-point failures, increased computational cost, extensive communication network costs and delays \cite{selforga}. The delays can hamper a fast-responding control action from the WF and can also compromise its tracking performance when (in a given set-up) the dispatching of the WGs has to happen fast e.g for maintaining power balance in a microgrid.
\par The literature on distributed methods for addressing the problem of dispatching distributively a group of WGs given a varying WF committed power output is not very broad. As far as the authors are aware, the only references related to the above problem are \cite{fullydistribdfig} and \cite{biegel}. In \cite{biegel}, the authors proposed  a distributed WF controller for regulating the power references of multiple WGs. At the same time, the controller ensured that the fatigue experienced by the WGs was reduced while the total power of the WF was reaching a pre-assigned value. In a similar line of research, the authors in \cite{fullydistribdfig} proposed a multi-agent-based strategy for addressing the same problem in a microgrid setting. 
The global information of the total demand and total available wind power were retrieved via a consensus protocol that was executed by the agents. Subsequently, this information was used by each agent to define the set point of its corresponding WG. 
\par In this work we make several contributions toward addressing the problem of distributively and dynamically dispatching a group of WGs for the purpose of having the WF power output tracking a reference. To this end, we first propose a \textit{distributed leader-follower consensus protocol} that realizes two basic control objectives 1) asymptotic consensus of the utilization levels of all WGs i.e consensus on the ratio defined by the available (from the wind) mechanical power over the maximum mechanical power of each WG 2)  asymptotic tracking of a varying reference by the         WF total power output. We prove that the proposed protocol asymptotically converges to its equilibrium point under specific conditions. Our proof relies on results from singular  perturbation theory \cite{khalil}. In the last part of our approach, we develop a Control Lyapunov Function-based (CLF) \cite{sontag}  RSC controller that realizes the proposed protocol in practice.
\par The rest of the paper is outlined in the following way. Section II describes the problem of distributed dynamic dispatching of the WGs for WF power output tracking. Section III, presents the relevant WF model. Section IV and V provide the main results of the paper. Section IV, introduces the proposed protocol and Section V presents the stability analysis.  Section VI gives the derivation of  the CLF-based torque RSC controller. In Section VII, the effectiveness of the proposed approach is evaluated via  simulations on the modified IEEE-RTS 24-bus system. Finally, Section VIII concludes the paper.

\section{Problem Formulation}
\subsection{Notation}

With $\mathcal{G}$ being a set we use $\left\vert\mathcal{G}\right\vert$ to denote its cardinality. We denote by $\mathbb{R}$  the set of reals and by $\mathbb{C}$ the set of complex numbers. Also, we denote by $\mathbb{R}_{+}$ the set of non-negative real numbers and with $\mathbb{R}_{++}$ the set of positive reals. We denote the $m$-dimensional Euclidean space  by $\mathbb{R}^m$. We denote vectors and matrices with bold characters. Let $\matr{A}\in\mathbb{R}^{m\times n}$ be a $m\times n$ matrix of reals. With $\matr{A}^\top$ we denote the transpose of $\matr{A}$ and with $[a]_{ij}$ the $(i,j)$-entry of the matrix $\matr{A}$. Let $\matr{A}\succ0\; (\matr{A}\succeq 0)$ denote that the matrix $\matr{A}$  is positive definite (semi-definite). The spectrum  of the matrix $\matr{A}$ (set of eigenvalues) is denoted by $\sigma(\matr{A})$. A $n\times n$ diagonal matrix $\matr{B}$ is denoted by $\matr{B}=\diag[b_i]_{i=1}^n$.  The maximum value of the vector $\matr{a}$ is denoted by $\bar{\matr{a}}$. 
  Similarly the maximum value of a scalar quantity $z$ is given by $\bar{z}$.  With $\matr{I}_n$ we denote the $n\times n$ identity matrix and with $\matr{0}_{n\times 1}$ and $\matr{1}_{n\times 1}$ a $n\times 1$ column vector of zeros and ones respectively.
 With $\dot{x}$ we denote the time derivative of a variable $x$ with respect to t, $\frac{dx}{dt}$.  With $\ddot{x}$ we denote the second derivative $\frac{d^2 x}{dt^2}$ The operator $\operatorname{Re}(\cdot)$ returns the real part of an imaginary number $(\cdot)\in\mathbb{C}$. With $\mathcal{C}^2$ we denote the space of functions with continuous second derivatives.

\subsection{Fair Dynamic Dispatching of WGs  While the WF Power Output is Tracking a Reference}
To formulate the main problem, we consider a set-up where we have a WF comprised with $n$ wind generators. We denote these generators by the set $\mathcal{G}\triangleq\{1,...,n\}$ and index each WG by $i$ where $i\in\mathcal{G}$. The available mechanical power that can be extracted from the wind by each WG is given by \cite{Pai}:
\begin{equation}
P_{m,i}\triangleq\frac{1}{2}\rho C_{p,i} A_{i} v_{w,i}^3\;, \hspace{10mm} \forall i\in\mathcal{G}
\label{mechanical power}
\end{equation}
where $\rho\in\mathbb{R}_{++}$ is the air density $[\frac{kg}{m^3}]$, $C_{p,i}\in\mathbb{R}_{+}$ is the power coefficient,  $A_i=\pi R_i^2\in\mathbb{R}_{++}$ is the area (swept by the blades) and $v_{w,i}\in\mathbb{R}_{++}$, the local wind speed in $[\frac{m}{s}]$. Notice that, the only controllable variable in \eqref{mechanical power} is $C_{p,i}$  which can be regulated by the rotor speed of the WG,  $\omega_{r,i}$. The standard functionality provided by DFIGs wind turbines is Maximum Power Point Tracking (MPPT). Achieving MPPT presumes that the WG is controlled such that  $C_{p,i}=\bar{C}_{p,i}$, where $\bar{C}_{p,i}$ is the maximum value of $C_{p,i}$. In that case, \eqref{mechanical power} can be recasted to:
\begin{equation}
\bar{P}_{m,i}\triangleq\frac{1}{2}\rho \bar{C}_{p,i} A_{i} v_{w,i}^3\;, \hspace{10mm} \forall i\in\mathcal{G}
\label{mechanical power max}
\end{equation}
The total power that a WF is required to extract from the wind at any given moment can be described by a reference $P_d$ that equals the WF committed power output to the grid. The latter is true under the mild assumption that the WF power losses are negligible. In our case, we consider a setting where the  WGs are operating in a deloading regime and can always meet the demanded power reference i.e $P_d\leq \sum_{i\in\mathcal{G}} \bar{P}_{m,i}$. For that setting, the problem we aim to address can be formulated as follows.
\begin{problem}
To develop a fully distributed control scheme for the RSC of the WGs that guarantees meeting the next two conditions.
\label{Problem1}
\end{problem}
\begin{cond}
$\lim\limits_{t\to\infty} \sum\limits_{i\in\mathcal{G}} P_{m,i}= P_d$
\label{cond1}
\end{cond}
\begin{cond}
$\lim\limits_{t\to\infty}\Big(\frac{P_{m,i}}{\bar{P}_{m,i}}\Big)=\lim\limits_{t\to\infty}\Big(\frac{P_{m,j}}{\bar{P}_{m,j}}\Big),\;\forall i,j\in \mathcal{G}$
\label{cond2}
 \end{cond}
The first condition ensures that the total WF power extracted from the wind is tracking the reference while the second condition that the WGs are dynamically dispatched in a fairly manner i.e the ratio of the mechanical power to the maximum mechanical power of each WG is the same. We suppose that all WGs have identical technical characteristics such that $A_i=A_j$, $\forall i,j \in\mathcal{G}$. The following remark holds. \vspace{2mm} 
\begin{remark}
$\Big(\frac{P_{m,i}}{\bar{P}_{m,i}}\Big)=\Big(\frac{C_{p,i}}{\bar{C}_{p,i}}\Big),\hspace{13mm}    \forall i \in \mathcal{G}$
\label{remark1}
\end{remark}
Remark~\ref{remark1} directly appears when dividing \eqref{mechanical power} over \eqref{mechanical power max}. 
With this, the Condition~\ref{cond2} becomes:
\begin{cond}
\begin{align*}
\lim\limits_{t\to\infty}\Big(\frac{C_{p,i}}{\bar{C}_{p,i}}\Big)=\lim\limits_{t\to\infty}\Big(\frac{C_{p,j}}{\bar{C}_{p,j}}\Big),\qquad\forall i,j \in \mathcal{G}
\end{align*}
\label{cond3}
\end{cond}
Observe that:
\begin{center}
Condition~\ref{cond2}$\Leftrightarrow$ Condition~\ref{cond3}
\end{center} 
In the sequel, we use this observation to introduce an approach that addresses \textit{Problem}~\ref{Problem1}.


\section{Mathematical Modeling}
We present the WF-related models for providing the ground of the forthcoming analysis.  Specifically, we present the wind-speed stochastic model, the rotor-voltage dynamical model including the RSC control input and the rotor-speed dynamical model.
\subsection{Wind Speed Model}
The effective wind speed $v_{w,i}\in \mathbb{R}_{+}$ can be modeled by integrating two basic components, the slowly-varying mean wind-speed, $v_{m,i}\in \mathbb{R}_{+}$, and the fast turbulence, $v_{s,i}\in \mathbb{R}_{+}$ \cite{larsen},\cite{thomsen}. Therefore, $v_{w,i}$ appears as:
\begin{equation}
v_{w,i}=v_{m,i}+v_{s,i},\qquad i\in\mathcal{G}
\end{equation}
 
We note that, the turbulent component can be modeled in the standard state-space form parameterized by the mean wind-speed $v_{m,i}$ with 
$p_{1,i}\triangleq p_{1,i}(v_{m,i}),\;p_{2,i}\triangleq p_{2,i}(v_{m,i}),\;k_i\triangleq k_i(v_{m,i})\in\mathbb{R}$ as:
\arraycolsep=1pt
\begin{equation}
\label{stochasticwind}
  \begin{pmatrix}
   \dot{v}_{s,i} \\
    \ddot{v}_{s,i}   \end{pmatrix} \triangleq
  \begin{pmatrix}
   0 & 1  \\
    -\frac{1}{p_{1,i}p_{2,i}}  &  -\frac{p_{1,i}+p_{2,i}}{p_{1,i}p_{2,i}}   \end{pmatrix}\begin{pmatrix}
   v_{s,i} \\
    \dot{v}_{s,i}   \end{pmatrix}+\begin{pmatrix}
   0 \\
   \frac{k_i}{p_{1,i}p_{2,i}}   \end{pmatrix}e
 \end{equation}
 With $e\in \mathcal{N}(0,1)$, we denote a white noise process \cite{larsen},\cite{thomsen}.
 \subsection{Wind Generator Model}
Since the problem we aim to address involves only the rotor of the WG, we only present the rotor-side dynamics of the WG. These, can be fully described by the electromagnetic state-variables (rotor-voltages) dynamics  and the electromechanical state-variable (rotor-speed) dynamics  \cite{Pai},\cite{dynamicmodeling} as:
\begin{subequations}
\thickmuskip=1mu
\begin{align}
\intertext{\textit{Rotor-voltage Dynamics}}
\dot{E}_{d,i}^'&\triangleq\frac{1}{T_{0,i}^'}\Big[-(E_{d,i}^'-(X_{s,i}-X_{s,i}^')I_{qs,i})\nonumber\\&+T_{0,i}^'(-\omega_s\frac{X_{m,i}}{X_{r,i}}V_{qr,i}+(\omega_s-\omega_{r,i})E_{q,i}^')\Big],\;\forall i\in\mathcal{G}\label{eddynamics}\\
\dot{E}_{q,i}^'&\triangleq\frac{1}{T_{0,i}^'}\Big[-(E_{q,i}^'+(X_{s,i}-X_{s,i}^')I_{ds,i})\nonumber\\&+T_{0,i}^'(\omega_s\frac{X_{m,i}}{X_{r,i}}V_{dr,i}-(\omega_s-\omega_{r,i})E_{d,i}^')\Big],\;\forall i\in\mathcal{G}\label{eqdynamics}
\intertext{\textit{Rotor-speed Dynamics}}
\dot{\omega}_{r,i} &\triangleq\frac{\omega_s}{2H_i}(T_{m,i}-T_{e,i})\;,\hspace{2mm}\forall i\in\mathcal{G}\label{rotorspeed}
\intertext{\textit{Mechanical Torque}}
T_{m,i}&\triangleq\frac{1}{2}\frac{\rho \pi  R_i^2 \omega_s}{S_{b,i} \omega_{r,i}}C_p(\lambda_i,\theta_i)v_{w,i}^3\;,\hspace{2mm}\forall i\in\mathcal{G}\label{mechanicaltorque}
\intertext{\textit{Power Coefficient}}
C_{p,i}(\lambda_i,\theta_i)&\triangleq 0.22\Big[116(\frac{1}{\lambda_i+0.08\theta_i}-\frac{0.035}{\theta_i^3+1})\Big]\nonumber \\
 &\cdot \mathrm{e}^{\left(-12.5(\frac{1}{\lambda_i+0.08\theta_i}-\frac{0.035}{\theta_i^3+1})\right)},\;\;\;\forall i\in\mathcal{G}\label{Cp}
\intertext{\textit{Tip-speed Ratio}}
\lambda_i &\triangleq\Big(\frac{2k_i}{p_i}\Big)\Big(\frac{\omega_rR_i}{v_{w,i}}\Big)\;,\;\;\forall i\in\mathcal{G}
\label{lambdai}
\end{align}
\end{subequations}
 All the variables are explained in the Appendix.

\section{Leader-Follower Consensus Protocol}
By establishing the equivalence between Condition \ref{cond2} and \ref{cond3}, we readily observe that
 having WGs achieving fair sharing is the same as ensuring that Condition \ref{cond3} is met. Consequently, we can pose this problem as a consensus agreement problem among all WGs on the utilization levels $\frac{C_{p,i}}{\bar{C}_{p,i}},\forall i\in\mathcal{G}$. For this reason, we introduce an appropriate \textit{Leader-follower Consensus Protocol}  which we prove that it converges to an equilibrium point that solves the exact consensus agreement problem. Let WG 1 be the group leader $l\triangleq 1$ and $\bar{\mathcal{G}}\triangleq\{2,...,n\}$ where $\bar{\mathcal{G}}\subset \mathcal{G}$  be the set of WGs without the leader. Then, we propose the following protocol.\\
\textbf{Protocol $\mathcal{P}_1$}
\begin{subequations}
\begin{align}
\intertext{\textit{Leader WG}}
&\frac{d\xi_h}{dt}\triangleq(P_d-P_{m,l}-\sum_{i\in\bar{\mathcal{G}}} P_{m,i})\hspace{11mm}\xi_h\in\mathbb{R}\label{consensus_equations1}\\
&\frac{dz_l}{dt}\triangleq-k_{\alpha,l}(z_l-\xi_h)\;,\;\;\;\; z_l\triangleq z_1 \hspace{9mm}z_l\in\mathbb{R}\label{consensus_equations2}\\
\intertext{\textit{WG} \;$i$}
&\frac{dz_i}{dt}\triangleq-k_{\alpha,i}(z_i-z_{i-1})\;,\;\;i\in\bar{\mathcal{G}} \hspace{9mm}z_i\in\mathbb{R}\label{consensus_equations3}
\end{align}
\end{subequations}
where the consensus states are $z_l\triangleq\Big(\frac{C_{p,l}}{\bar{C}_{p,l}}\Big)$ and $z_i\triangleq\Big(\frac{C_{p,i}}{\bar{C}_{p,i}}\Big)$ respectively and the auxiliary state-variable of the leader is  $\xi_h$. The allowable communication links for implementing this protocol can be seen in Fig.~\ref{physcomtop}.  
We briefly describe the mechanism by which the protocol is executed. The WF supervisory controller obtains the WF reference $P_d$ from the system operator and passes its value to the leader WG. Next, the leader WG computes its auxiliary state $\xi_h$ and its consensus state $z_l$, using the reference and information from all WGs. The leader communicates its consensus state with the time derivative of its consensus state to the neighboring WG. The same process is executed by all WGs i.e they obtain their consensus state and communicate it to a neighboring WG, concurrently.  An assumption that has to be valid in order for the protocol to be realizable is that the leader WG can retrieve the information $\sum_{i\in\bar{\mathcal{G}}}P_{m,i}$. The following methods can be used for this purpose.
\begin{enumerate}
\item information passing from each WG to the leader (indirectly, via intermediate WGs).
 \item  average consensus protocol \cite{murray} with consensus state-variable the mechanical power i.e $\sum_{i\in\bar{\mathcal{G}}} (P_{m,i}+P_{m,l})=(P_{m,avg}\cdot n)$.
\end{enumerate}

\section{Stability Analysis of the Proposed Protocol}
In this section, we study the asymptotic behavior of the proposed protocol and the properties of its equilibrium point. We begin by defining the following coefficients:
\begin{align}
\alpha_l&\triangleq\frac{1}{2}\rho \bar{C}_{p,l}  A_{l} v_{w,l}^3,\;\;\alpha_l \in\mathbb{R}_{+} \\
 \alpha_i&\triangleq\frac{1}{2}\rho \bar{C}_{p,i} A_{i} v_{w,i}^3 ,\;\;\alpha_i \in\mathbb{R}_{+},\;i\in\bar{\mathcal{G}}
\end{align}
 where in vector form are written as $\pmb{\alpha}=[\alpha_l,\alpha_2, ... ,\alpha_{n}]^\top, \pmb{\alpha} \in \mathbb{R}^n$.  We define the consensus-states vector compactly as $\matr{z}=[z_1,... ,z_i,..., z_n]^\top,\;\; \matr{z}\in\mathbb{R}^n$. With these, Eq.~\eqref{consensus_equations1} become:
\begin{subequations}
\begin{align}
\frac{d\xi_h}{dt}&=(P_d-\alpha_lz_l-\sum_{i\in\bar{\mathcal{G}}} \alpha_i z_i)\label{consensus_equations4}
\end{align}
\end{subequations}
The equilibrium of the consensus protocol as obtained from Eq.~\eqref{consensus_equations4},\eqref{consensus_equations2},\eqref{consensus_equations3} is:
\begin{align}
 \xi_{h0}&=\frac{P_d}{(\alpha_l+\sum_{i\in\bar{\mathcal{G}}}\alpha_i)}\\
z_{l0}&=\xi_{h0}\\
z_{i0}&=\xi_{h0},\;\; \forall i\in\bar{\mathcal{G}}
\end{align} 
Without loss of generality we take $k_{\alpha,i}\triangleq k_{\alpha,l},\;\forall i\in \bar{\mathcal{G}}$. Defining $\varepsilon\in\mathbb{R}_{+}$ as $\varepsilon=\frac{\bar{\alpha}}{k_{a,i}}$  we have the next Theorem.
\begin{theorem}
$\exists\varepsilon^*>0$ s.t  $\forall\varepsilon<\varepsilon^*$ the equilibrium point $(\xi_{h0},\matr{z}_0)$ is \textit{asymptotically stable}.
\end{theorem}
\begin{proof}
 First, we define a new time-scale $\tau=\bar{\alpha}\; t$ with $d\tau=\bar{\alpha}\;dt$. Using this, equations \eqref{consensus_equations2},\eqref{consensus_equations3}, \eqref{consensus_equations4} become:
\begin{subequations}
\begin{align}
\frac{d{\xi}_h}{d\tau}&\triangleq(\frac{P_d}{\bar{\alpha}}-\frac{\alpha_l}{\bar{\alpha}} z_l-\sum_{i\in \bar{\mathcal{G}}} \frac{\alpha_i}{\bar{\alpha}} z_i)
\label{consensus_equations5}\\
\Big(\frac{\bar{\alpha}}{k_{\alpha,l}}\Big)\frac{d{z}_l}{d\tau}&\triangleq-(z_l-\xi_h)
\label{consensus_equations6}\\
\Big(\frac{\bar{\alpha}}{k_{\alpha,i}}\Big)\frac{d{z}_i}{d\tau}&\triangleq-(z_i-z_{i-1}),\qquad i\in\mathcal{\bar{G}}
\label{consensus_equations7}
\end{align}
\end{subequations}
 Letting  $\frac{\bar{a}}{k_{\alpha,l}}=\frac{\bar{a}}{k_{\alpha,l}}=\varepsilon$, we can write the above equations more compactly as:\\
 \textit{Slow quasi-steady system}
\begin{subequations}
\begin{align}
\frac{d{\xi}_h}{d\tau}&\triangleq g_h,\hspace{5mm} g_h\in\mathbb{R}
\label{consensus_equations8}
\intertext{\textit{Fast boundary-layer system}}
\varepsilon \frac{\matr{dz}}{\matr{d}\bm{\tau}}&\triangleq\matr{g},\hspace{5mm} \matr{g}\in\mathbb{R}^n
\label{consensus_equations10}
\end{align}
where
\begin{align}
 g_h&\triangleq(\frac{P_d}{\bar{\alpha}}-\frac{\alpha_l}{\bar{\alpha}} z_l-\sum_{i\in \bar{\mathcal{G}}} \frac{\alpha_i}{\bar{\alpha}} z_i)\nonumber \\
 \matr{g}&\triangleq [-(z_l-\xi_h) \dots -(z_i-z_{i-1}) \dots -(z_n-z_{n-1})]^\top
\end{align}
\end{subequations}
 Equations \eqref{consensus_equations8},\eqref{consensus_equations10} are in the \textit{standard singularly perturbed form} \cite{khalil} with $\xi_h$  being the slow state-variable and $\matr{z}$ being the fast state-variables. A system possessing a multi-time-scale property enables a compartmental  stability analysis of its system dynamics. By exploiting this property, we first study the fast boundary-layer system dynamics \eqref{consensus_equations10} in a new time scale $\tilde{\tau}=\frac{\tau}{\varepsilon}$. Assuming that the slow state-variable $\xi_h$ is \textit{``frozen''} i.e $\frac{d\xi_h}{d\tilde{\tau}}\approx 0$, and using the transformation $y_i=(z_i-\xi_h)$ we can write the system equations \eqref{consensus_equations10} as:
\begin{align}
\frac{\matr{dy}}{\matr{d}\bm{\tilde{\tau}}}&\triangleq\matr{A}_f \matr{y},
\label{fastsystem}\\
\matr{y}&\triangleq[y_1,...,y_n]^\top\\
\matr{A}_f&\triangleq\underbrace{\begin{bmatrix}
-1 & 0 & \cdots  & 0 & 0 \\
1 & -1 & \cdots   & 0 & 0 \\
\vdots & \ddots     &  &  & \vdots  \\
0 & 0 & \cdots  & 1 & -1 \\
\end{bmatrix}}_{n\times n}
\end{align}
Since $\matr{A}_f$ is a lower-triangular matrix, the diagonal terms represent also the eigenvalues of $\matr{A}_f$. From that, we can conclude that  $\matr{A}_f$ is a \textit{Hurwitz matrix} and that the equilibrium $\matr{y}_0=\matr{0}_{n\times 1}$ is \textit{asymptotically stable}. Equivalently that, the equilibrium $\matr{z}_0=(\xi_{h0}\cdot \matr{1}_{n\times 1})$ is \textit{asymptotically stable} and attractive to the trajectories of the fast state-variables $\matr{z}$. Thus, all $z_i$ will converge toward the slow state-variable $\xi_h$. We are left to show that the $\xi_h$ converges toward $\xi_{h0}$. To do that, we focus on the behavior of the fast sub-system equation \eqref{consensus_equations10} when $\varepsilon\triangleq 0$, and observe that it degenerates into the algrebraic equation:
\begin{align}
\matr{0}_{n\times 1}=\matr{g}(\xi_h,\matr{z})
\label{slowmanifold}
\end{align}
Solving for $\matr{z}$ results into the  $n-$dimensional equilibrium slow-manifold of \eqref{slowmanifold}, described by $\matr{z}=\eta(\xi_h)=(\xi_h\cdot \matr{1}_{n\times 1})$. Direct substitution into \eqref{consensus_equations8} yields the slow model:
\begin{align}
\frac{d{\xi}_h}{d\tau}&=g_h(\xi_h,\eta(\xi_h))\nonumber\\&=\Big(\frac{P_d}{\bar{\alpha}}-\frac{\alpha_l}{\bar{\alpha}} \xi_h-\sum_{i\in \bar{\mathcal{G}}} \frac{\alpha_i}{\bar{\alpha}} \xi_h\Big)
\label{consensus_equations8_new}
\end{align}
The slow sub-system \eqref{consensus_equations8_new} has an \textit{asymptotically stable} equilibrium equal to:
 \begin{align*}
 \xi_{h0}=\frac{P_d}{(\alpha_l+\sum_{i\in\bar{\mathcal{G}}}\alpha_i)}
 \end{align*}
   Having established that, the fast and the slow sub-systems have asymptotically stable equilibria, from Theorem 11.4 (in \cite{khalil}) we conclude the next statement. That, $\exists\varepsilon^*>0$ such that $\forall\varepsilon<\varepsilon^*$ the equilibrium point of the full system \eqref{consensus_equations5}-\eqref{consensus_equations7} is \textit{asymptotically stable}. That, completes the proof.\end{proof}

\section{Design of the Consensus-based Torque Controller}
The previous Section was dedicated to establishing asymptotic convergence of the proposed consensus protocol to an equilibrium point that realizes the desired control objectives, described in \textit{Problem} \ref{Problem1}.  For implementing the proposed protocol, we  design a cooperative torque controller for the RSC of each WG that will force the system dynamics to evolve as in \eqref{consensus_equations5}-\eqref{consensus_equations7}.  Writing equation~\eqref{consensus_equations3} analytically leads to:
\begin{equation}
\frac{\dot{C}_{p,i}}{\bar{C}_{p,i}}=-k_{\alpha,i}(\frac{C_{p,i}}{\bar{C}_{p,i}}-\frac{C_{p,i-1}}{\bar{C}_{p,i-1}}),\hspace{20mm} i\in\mathcal{G}
\label{cpiconsensus}
\end{equation}
Noticing from \eqref{Cp} that when $\theta_i=0$  we have
\begin{align}
C_{p,i}\triangleq C_{p,i}(\lambda_i,0),\hspace{38mm} i\in\mathcal{G}
\end{align} The term $\frac{\dot{C}_{p,i}}{\bar{C}_{p,i}}$ in \eqref{cpiconsensus} can be expanded as:
\begin{equation}
\frac{\dot{C}_{p,i}}{\bar{C}_{p,i}}\triangleq\Big(\frac{1}{\bar{C}_{p,i}}\Big)\Big(\frac{\partial C_{p,i}}{\partial \lambda_i}\Big)\Big(\frac{\partial \lambda_i}{\partial \omega_{r,i}}\Big)\Big(\frac{\omega_s}{2H_i}\Big)(T_{m,i}-T_{e,i}),\hspace{3mm} i\in\mathcal{G}\label{cpidot}
\end{equation}
Letting  ~\eqref{cpidot} and \eqref{cpiconsensus} to be equal gives the electrical torque $T_{e,i}^{*}\in\mathbb{R}$  as:
\begin{equation}
T_{e,i}^{*}=T_{m,i}-(\frac{1}{\bar{C}_{p,i}}\frac{\partial C_{p,i}}{\partial \lambda_i}\frac{\partial \lambda_i}{\partial \omega_{r,i}}\frac{\omega_s}{2H_i})^{-1}\Big[-k_{\alpha,i}(\frac{C_{p,i}}{\bar{C}_{p,i}}-\frac{C_{p,i-1}}{\bar{C}_{p,i-1}})\Big]
\end{equation}
To guarantee that $\lim_{t\to\infty} (T_{e,i})=T_{e,i}^{*}$ we consider the candidate Control Lyapunov Function (CLF):
\begin{equation}
V_{e,i}=\frac{1}{2}(T_{e,i}-T_{ei}^{*})^2,\qquad  \forall i\in\mathcal{G}
\end{equation}
where $V_{e,i}>0,\;\forall T_{e,i}\in \mathcal{D}_{e,i}\setminus\{T_{ei}^{*}\} $.
Now, consider the following proposition.
\begin{proposition}
$V_{e,i}\in \mathcal{C}^2$ is a \textit{CLF} and the equilibrium $T_{e,i}=T_{e,i}^*$ can be rendered \textit{asymptotically stable}.
\end{proposition}
\begin{proof}
Define the variable $\tilde{T}_{e,i}=(T_{e,i}-T_{e,i}^*),\;\;\forall i\in\mathcal{G}$ and let the electrical torque expressed as:
\begin{equation}
T_{e,i}=\frac{E_{q,i}^{'}V_{s,i}}{X_{s,i}^{'}},\;\;  \forall i\in\mathcal{G}
\label{Tei}
\end{equation}
 Computing the time-derivative of $V_{e,i}$ along $\dot{\tilde{T}}_{e,i}$ dynamics gives us:
\begin{align}
\frac{dV_{e,i}}{dt}&=\tilde{T}_{e,i} \frac{V_{s,i}}{X_{s,i}^'}\frac{1}{T_{0,i}^'}\Big[-(E_{q,i}^'+(X_{s,i}-X_{s,i}^')I_{ds,i})\nonumber\\&+T_{0,i}^'(\omega_s\frac{X_{m,i}}{X_{r,i}}V_{dr,i}-(\omega_s-\omega_{r,i})E_{d,i}^')\Big]-\tilde{T}_{e,i}\dot{T}_{e,i}^{*}
\end{align}
This expression can be compactly expressed as:
\begin{equation}
\frac{dV_{e,i}}{dt}=\frac{\partial V_{e,i}}{\partial\tilde{T}_{e,i}} [f_i+h_iV_{dr,i}],\qquad \forall i\in\mathcal{G}
\end{equation}
where 
\begin{align}
\frac{\partial V_{e,i}}{\partial\tilde{T}_{e,i}} f_i&\triangleq\tilde{T}_{e,i} \frac{V_{s,i}}{X_{s,i}^'}\frac{1}{T_{0,i}^'}\Big[-(E_{q,i}^'+(X_{s,i}-X_{s,i}^')I_{ds,i})\nonumber\\&+T_{0,i}^'(-(\omega_s-\omega_{r,i})E_{d,i}^')\Big]-\tilde{T}_{e,i}\dot{T}_{e,i}^{*},\qquad \forall i\in\mathcal{G}\\
\frac{\partial V_{e,i}}{\partial\tilde{T}_{e,i}}h_{i}&\triangleq\tilde{T}_{e,i} \frac{V_{s,i}}{X_{s,i}^'}\omega_s\frac{X_{m,i}}{X_{r,i}},\qquad \forall i\in\mathcal{G}
\end{align}
We observe that $\frac{\partial V_{e,i}}{\partial\tilde{T_{e,i}}}h_i\neq 0$ whenever $\tilde{T}_{e,i}\neq 0$. That, means the feedback control input $V_{dr,i}$ can always guarantee that $\frac{\partial V_{e,i}}{\partial\tilde{T_{e,i}}}f_i< 0,\;\;$  $\forall\tilde{T}_{e,i}\neq 0$. That, proves that $V_{e,i}$ is indeed a CLF and that the equilibrium $T_{e,i}=\tilde{T}_{e,i}^{*}$ can be rendered exponentially stable. That, completes the proof.\end{proof}
We note that,   $V_{e,i}$ being a   CLF is a necessary condition for the existence of a stabilizing feedback control $V_{dr,i}$. Hence, we proceed by designing a RSC stabilizing controller $V_{dr,i}$. For having $\dot{V}_{e,i}<0$ we take $\dot{V}_{e,i}$ to be:
\begin{equation}
\frac{dV_{e,i}}{dt}=-k_{\beta,i}(T_{e,i}-T_{e,i}^{*})^2 <0, \;\;\;\;\;\forall T_{e,i}\in \mathcal{D}_{e,i}\setminus\{T_{e,i}^{*}\}
\label{Vedot}
\end{equation}
Equation \eqref{Vedot} can be written as $\frac{dV_{e,i}}{dt}=-2k_{\beta,i}(V_{e,i})$ which has solution $V_{e,i}=V_{e,i 0}\mathrm{e}^{-2k_{\beta,i}t}$ i.e $\lim_{t\to\infty}V_{e,i}=0,\;$ 
For having \eqref{Vedot} the following equation has to hold: \begin{equation}
\frac{dT_{e,i}}{dt}=\frac{dT_{e,i}^{*}}{dt}-k_{\beta,i}(T_{e,i}-T_{e,i}^{*}),\;\;  \forall i\in\mathcal{G}
\label{Teidot}
\end{equation}
Furthermore, we assume that  $\frac{dV_{s,i}}{dt}=0,\;\forall i\in\mathcal{G}$ and $\frac{\partial}{\partial t}(\frac{\partial \lambda_i}{\partial \omega_{r,i}})=0,\;\forall i\in\mathcal{G}$ hold.
\begin{figure}
\centering
\includegraphics[scale=0.49]{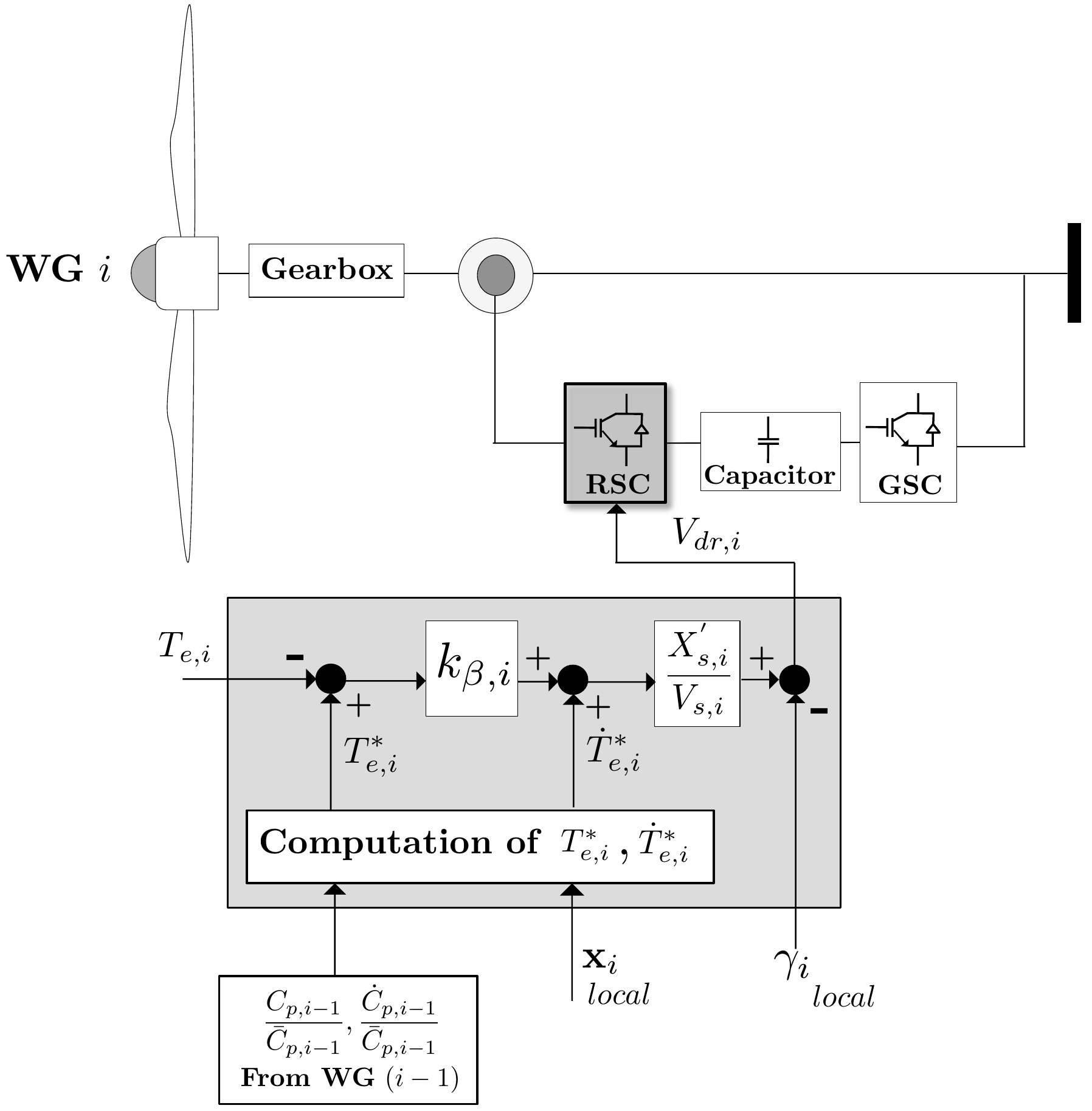}
\caption{Cooperative torque controller of WG $i$}
\label{CLFcontr}
\end{figure} 
 Combining equations~ \eqref{Teidot},\eqref{Tei},\eqref{eqdynamics} we derive the RSC controller as:
\begin{align}
\label{controlinput1}
V_{dr,i}&=\frac{X_{s,i}^{'}}{V_{s,i}}\Big[\frac{dT_{e,i}^{*}}{dt}-k_{\beta,i}(T_{e,i}-T_{ei}^{*})\Big]-\frac{1}{T_{0,i}^'}\Big[-(E_{q,i}^'\nonumber\\&+(X_{s,i}-X_{s,i}^')I_{ds,i})\nonumber\\&+T_{0,i}^'(-(\omega_s-\omega_{r,i})E_{d,i}^')\Big]\frac{X_{r,i}}{X_{m,i}\omega_s}, \hspace{5mm} i\in\mathcal{G}
\end{align}
This controller is depicted in Fig.~\ref{CLFcontr} with all the variables explained in the Appendix.
The expressions of the appearing terms $\frac{\partial \lambda_i}{\partial \omega_{r,i}},\frac{\partial C_{p,i}}{\partial \lambda_i},\frac{\partial^2 C_{p,i}}{\partial \lambda_i^2}$ are ommitted due to space limitation. Nevertheless, they can be derived from \eqref{Cp},\eqref{lambdai}.


\begin{figure}
\begin{center}
\includegraphics[scale=0.38]{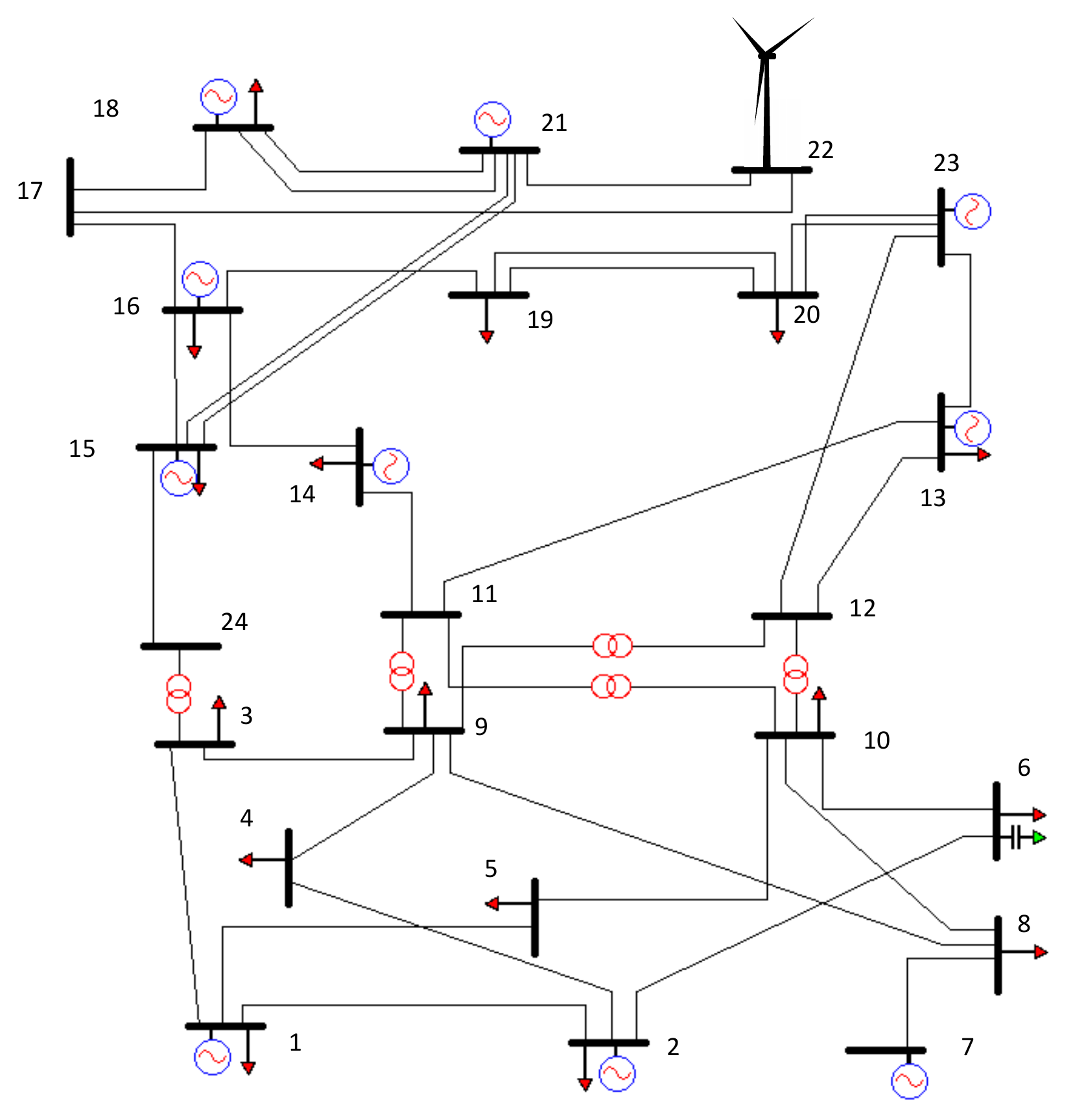}
\caption{IEEE 24-bus RT system with a WF (with 10 WGs) at bus 22}
\label{24buswind}
\end{center}
\end{figure}

\begin{figure}
       \begin{subfigure}{0.34\textwidth}
   \includegraphics[scale=0.34]{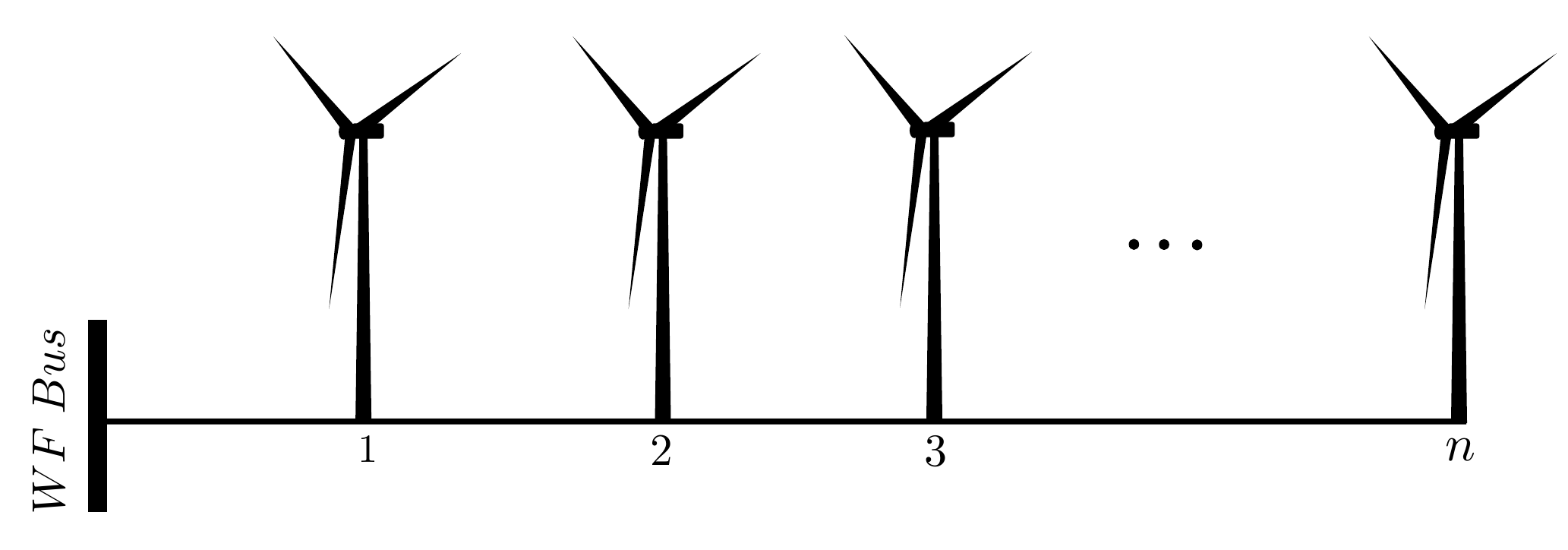}
        \caption{}
        \label{physicalcon}
    \end{subfigure}\\
    \begin{subfigure}{0.32\textwidth}    
   \includegraphics[scale=0.32]{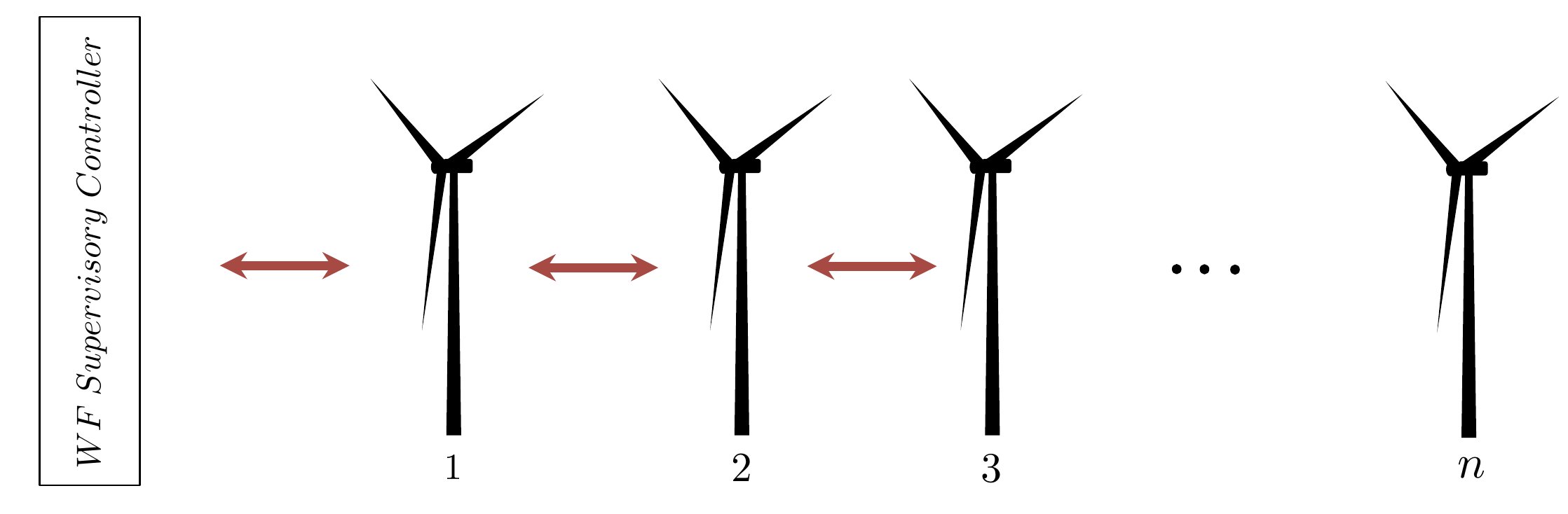}
        \caption{} 
        \label{commcon}
    \end{subfigure}
         \caption{a) Physical connectivity of the WF b) Communication structure between WGs }
        \label{physcomtop}
\end{figure}

\section{Performance Evaluation}
The effectiveness of the proposed approach is explored via simulations on the modified IEEE 24-bus RT system. In this system, a WF comprised with 10 WGs lies at bus 22.  The physical connectivity and the allowable communication links among the WGs (for $n=10$) are shown in Fig.~\ref{physcomtop}. The control logic for each RSC follows equation \eqref{controlinput1} (Fig.~\ref{CLFcontr}) whereas the group objective for the RSC controllers is to dynamically self-dispatch their WGs in a fair-sharing way and the WF power output to track a reference shown in Fig.~\ref{pwftrack}.  We studied the following critical scenarios.
\begin{align}
\textit{Scenario 1}:& \text{ At } t=0s, \text{ the WF power output reference is} \nonumber
\\[-0.05in]& \text{$0.38 p.u$  and suddenly at } t=0.2s, \text{the reference }\nonumber\\[-0.05in]& \text{changes step-wise to $0.42 p.u$ as seen in Fig.~\ref{pwftrack}}.\nonumber
\end{align}
The  response of the WF power output (blue) is tracking the reference (red) closely with good performance given standard metrics e.g overshoot, response time (Fig.~ \ref{pwftrack}).  The response of the consensus state-variables is depicted in Fig.~\ref{cptrack}.  Notice that, the response for all 10 WGs is completely identical. That, verifies the ``fair-dispatching'' between the WGs i.e each WG extracts mechanical power from the wind according to the local wind-speed conditions. In our setting, we regard that all WGs experience the same local wind-speed conditions. The mechanism by which the CLF-based RSC controller carries out its objectives is understood as follows. When the leader WG obtains the new power reference its consensus state-variable is ordered to increase value. Since all the WGs are trying to reach consensus with the leader they increase their consensus state variables, leading all the utilization factors to exceed $0.8$ while starting from a value of $0.73$ (Fig.~\ref{cptrack}).  To achieve that, the various RSC torque controllers slowed-down the WGs (Fig.~\ref{rotortrack}), enabling them to increase the mechanical power that they extract from the wind until their total power reached the pre-assigned reference value (Fig.~\ref{pwftrack}).  In summary,  the proposed protocol and the developed RSC controllers effectively  address  Problem~\ref{Problem1}.
\begin{figure}
\centering
  \begin{subfigure}{0.35\textwidth}
\includegraphics[scale=0.35]{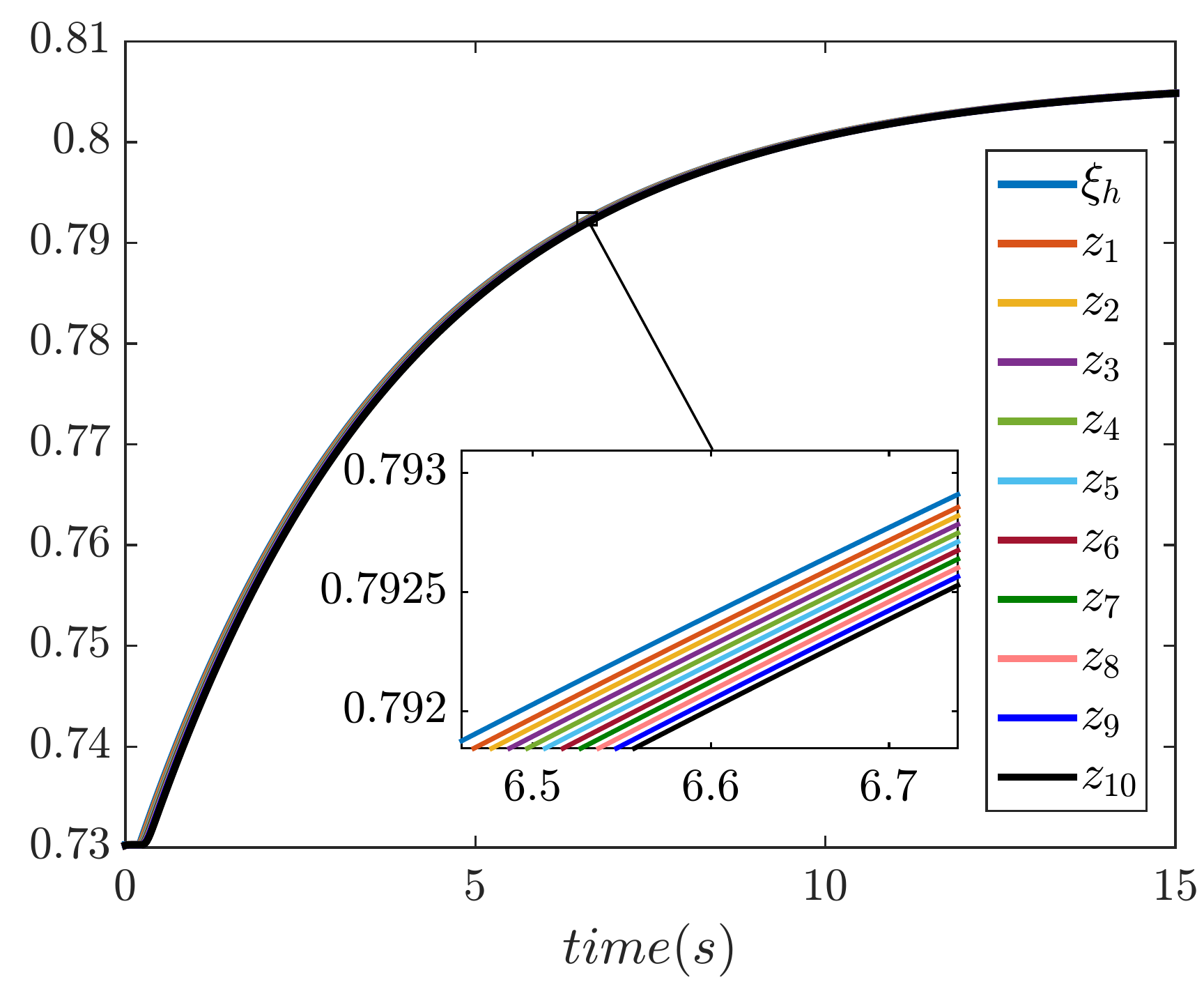}
\caption{Response of the $C_{p,i}$ coefficients }
\label{cptrack}
\end{subfigure}
\\
  \begin{subfigure}{0.35\textwidth}
\includegraphics[scale=0.35]{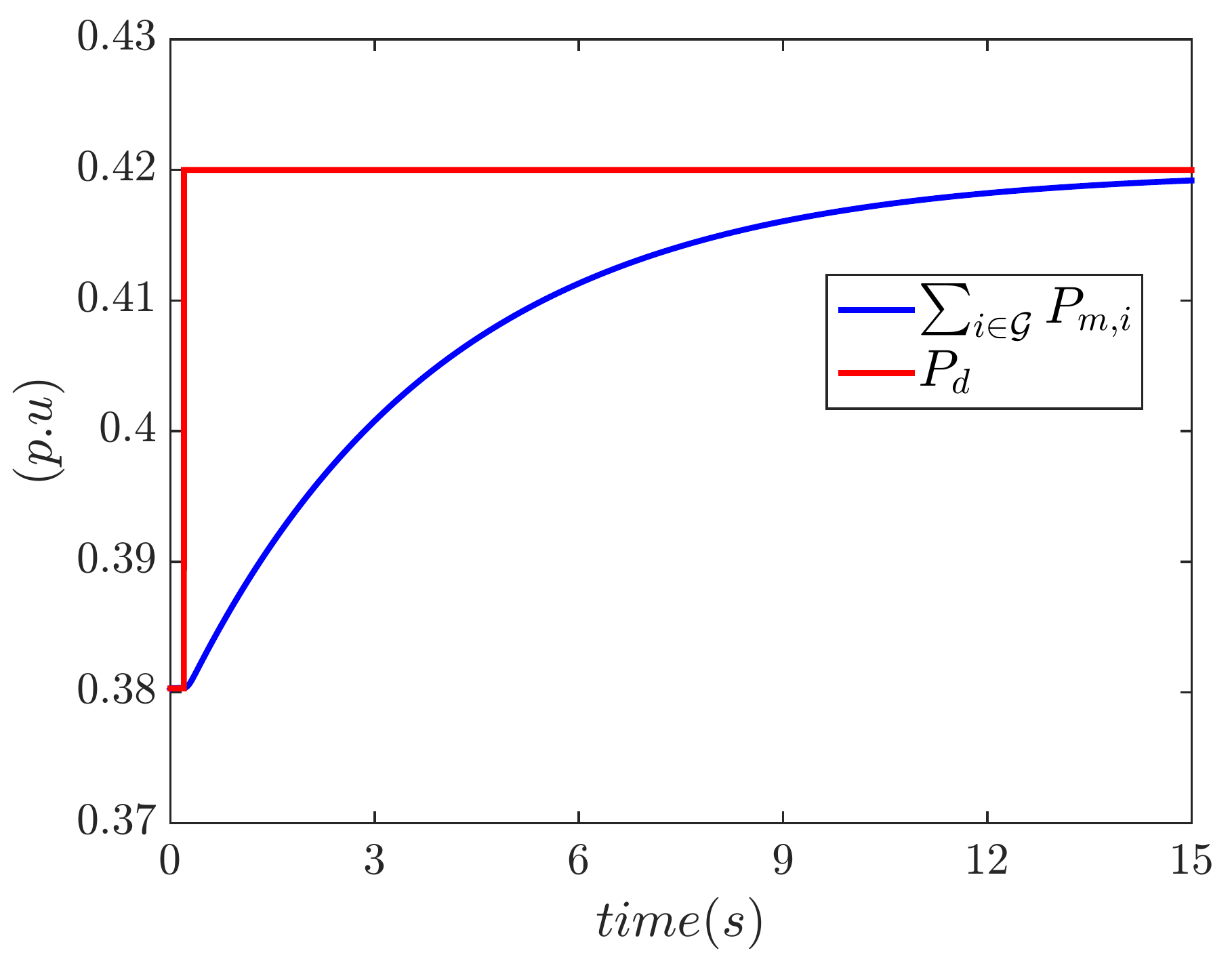}
\caption{WF total mechanical power tracking response (reference $P_d$)}
\label{pwftrack}
\end{subfigure}\\
  \begin{subfigure}{0.35\textwidth}
\includegraphics[scale=0.35]{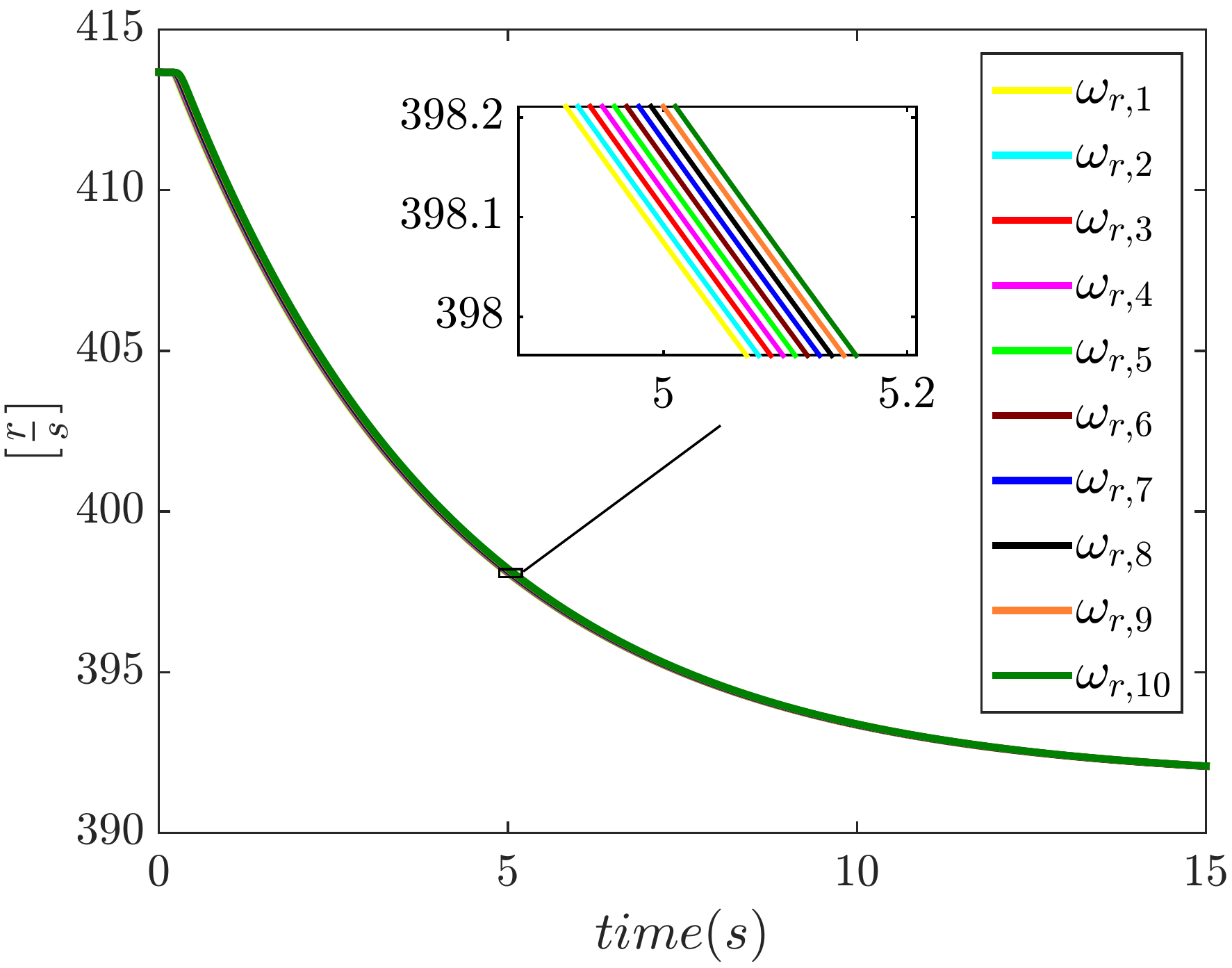}
\caption{Rotor speed response}
\label{rotortrack}
\end{subfigure}
\caption{System response under Scenario 1}
\end{figure}
\section{Concluding Remarks}
This paper introduced a leader-follower consensus protocol that is able to dynamically dispatch a fleet of WGs according to their local wind-speed conditions such that the WF total power output reaches a new assigned value. By employing singular perturbation theory \cite{khalil}, we provided theoretical guarantees in the form of asymptotic stability of desired equilibria, proving that the protocol will asymptotically accomplish its aforementioned objectives.  On the practical side, we developed a cooperative CLF-based RSC controller that implements the above protocol.  We demonstrated the performance of the proposed methodology via simulations on the IEEE 24-bus RT system.
\section{Appendix}
The various variables related to WG $i$ are explained below.\\\\
 \begin{tabular}{| l | l | }
\hline
\textbf{Variable} & \textbf{Corresponds to} \\\hline
 $T_{e,i}\in \mathbb{R}_{+}$ & electrical torque \\\hline
 $T_{m,i}\in \mathbb{R}_{+}$ &  mechanical torque \\\hline
 $T_{0,i}^'\in \mathbb{R}_{+}$ & transient open-circuit time constant\\\hline
 $X_{s,i}^'\in \mathbb{R}_{+}$ &  stator transient reactance\\\hline
 $X_{s,i}\in \mathbb{R}_{+}$ &  stator reactance \\\hline
 $X_{r,i}\in \mathbb{R}_{+}$ & rotor reactance\\\hline
 $X_{m,i}\in \mathbb{R}_{+}$ & mutual reactance of the stator-rotor\\\hline
 $H_i\in \mathbb{R}_{+}$ & combined inertia of the WG\\\hline
 $E_{q,i}^',E_{d,i}^'\in \mathbb{R}_{+}$ &  $q,d$ axis rotor voltages\\ 
 \hline
 $I_{qs,i},I_{ds,i}\in \mathbb{R}_{+}$ & $q,d$ axis stator current \\ \hline
   $V_{qr,i},V_{dr,i}\in \mathbb{R}_{+}$ &  $q$ and $d$ axis RSC control inputs \\\hline
  $\omega_s$ & synchronous speed $2\pi\cdot 60$ $[\frac{rad}{s}]$\\\hline
  $S_{b,i}\in \mathbb{R}_{+}$ & base power \\\hline
  $\omega_{r,i}\in \mathbb{R}_{+}$&  electrical rotor speed of the WG \\\hline
  $\lambda_i\in \mathbb{R}_{+}$ & tip speed ratio\\\hline
  $\theta_i\in \mathbb{R}$ &  pitch angle\\\hline
 $k_i\in \mathbb{R}_{+}$ &  gearbox ratio\\\hline
 $p_i\in \mathbb{R}_{+}$ &  poles\\\hline
 \end{tabular}
 
\begin{align}
\matr{x}_i&=[T_{m,i}\;,\dot{T}_{m,i}\;,\omega_{r,i}\;,\dot{\omega}_{r,i}\;,\frac{\partial C_{p,i}}{\partial\lambda_i}\; , \frac{\partial^2C_{p,i}}{\partial\lambda_i^2}\;,\frac{C_{p,i}}{\bar{C}_{p,i}}\;, \frac{\dot{C}_{p,i}}{\bar{C}_{p,i}},\nonumber\\
&\;\;\frac{\partial \lambda_i}{\partial \omega_{r,i}},\;v_{w,i}]^\top\\
\gamma_i&=\frac{1}{T_{0,i}^'}\Big[-(E_{q,i}^'+(X_{s,i}-X_{s,i}^')I_{ds,i})\nonumber\\&+T_{0,i}^'(-(\omega_s-\omega_{r,i})E_{d,i}^')\Big]\frac{X_{r,i}}{X_{m,i}\omega_s}\\
\dot{T}_{e,i}^{*}&=\dot{T}_{m,i}+k_{\beta,i}\bar{C}_{p,i}(\frac{\omega_s}{2H_i}\frac{\partial \lambda_i}{\partial \omega_{r,i}})^{-1}(\frac{\partial C_{p,i}}{\partial \lambda_i})^{-2}\nonumber\\
&\cdot\Big[(\frac{\dot{C}_{p,i}}{\bar{C}_{p,i}}-\frac{\dot{C}_{p,i-1}}{\bar{C}_{p,i-1}})\frac{\partial C_{p,i}}{\partial \lambda_i}
\nonumber\\
&-(\frac{C_{p,i}}{\bar{C}_{p,i}}-\frac{C_{p,i-1}}{\bar{C}_{p,i-1}})\frac{\partial^2 C_{p,i}}{\partial \lambda_i^2}\cdot\dot{\omega}_{r,i}\Big],\;\;\forall i\in\mathcal{G}
\end{align}

\nocite{*}
\IEEEpeerreviewmaketitle
\bibliographystyle{unsrt}
\bibliography{ACC2016_Distributed_torque_contr}{}

\end{document}